\documentclass{IEEEoj}
\usepackage{cite}
\usepackage{amsmath,amssymb,amsfonts}
\usepackage{graphicx,color}
\usepackage{textcomp}
\usepackage{algorithm}
\usepackage{algpseudocode}
\usepackage{multirow}  %
\usepackage{array}  %
\usepackage{subfig}
\usepackage{amsthm}
\newtheorem{proposition}{Proposition}
\newtheorem{remark}{Remark}

\def\BibTeX{{\rm B\kern-.05em{\sc i\kern-.025em b}\kern-.08em
    T\kern-.1667em\lower.7ex\hbox{E}\kern-.125emX}}
\AtBeginDocument{\definecolor{ojcolor}{cmyk}{0.93,0.59,0.15,0.02}}
\def\OJlogo{\vspace{-4pt}\hskip-4pt\includegraphics[height=18pt]{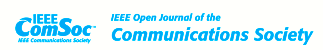}}
\begin{document}
\receiveddate{XX Month, XXXX}
\reviseddate{XX Month, XXXX}
\accepteddate{XX Month, XXXX}
\publisheddate{XX Month, XXXX}
\currentdate{11 January, 2024}
\doiinfo{OJCOMS.2024.011100}

\title{Joint Geometric and QoS-Aware Routing in Optical LEO Satellite Networks via DRL}
\author{%

Abdulrahman~Al-Hababi\IEEEauthorrefmark{1},
Meysam~Ghanbari\IEEEauthorrefmark{2},
Mohammad~Taghi~Dabiri\IEEEauthorrefmark{3},
Rula~Ammuri\IEEEauthorrefmark{4},
Mazen~Hasna\IEEEauthorrefmark{5}\IEEEmembership{(Senior Member, IEEE)},
Khalid~A.~Qaraqe\IEEEauthorrefmark{2}\IEEEmembership{(Senior Member, IEEE)}
}

\affil{
Texas A\&M University at Qatar, Doha, Qatar
}

\affil{
College of Science and Engineering, Hamad Bin Khalifa University, Doha, Qatarn
}

\affil{
Electrical Engineering Division, Department of Engineering, University of Cambridge, Cambridge, United Kingdom
}

\affil{
Professionals for Smart Technology (PST), Amman, Jordan
}

\affil{
Department of Electrical Engineering, Qatar University, Doha, Qatar
}

\corresp{Corresponding author: Abdulrahman~Al-Hababi (e-mail: aalhababi@tamu.edu).}

\authornote{This work was supported by the Qatar Research Development and Innovation Council (QRDI) under Grant No. NPRP14C-0909-210008 and by research funding from Hamad Bin Khalifa University under the Thematic Research Grant Program Cycle 3. The statements made herein are solely the responsibility of the authors. The content is solely the responsibility of the authors and does not necessarily represent the official views of QRDI.}

\markboth{Preparation of Papers for IEEE Open Journals}{Hababi \textit{et al.}}

\begin{abstract}
	Optical inter-satellite links (ISLs) are becoming the backbone of modern LEO constellations, offering high capacity and low latency but introducing stringent geometric and physical-layer constraints. Routing in such networks must therefore account for time-varying topology, jitter-induced outage, and the heterogeneous reliability of intra- and inter-plane optical links, aspects that classical shortest-path or existing learning-based schemes do not fully capture.
	This paper develops a joint geometric and QoS-aware routing framework for optical LEO networks. We derive a closed-form outage expression under Gaussian-beam propagation with pointing errors and obtain analytical maximum feasible link ranges for different ISL classes. These relations remove beam divergence from the optimization variables and embed optical feasibility directly into the routing layer, leading to a latency–reliability–capacity constrained routing formulation that is proved to be NP-hard.
	To enable scalable decision-making, we cast snapshot routing as a Markov decision process and introduce an angle-constrained masked deep Q-network (AC-MDQN) that integrates optical feasibility masks, potential-based latency shaping, and a geometry-aware corridor filter around the source–destination great-circle path. This design significantly reduces the effective action-space complexity while preserving near-optimal routing choices. Simulations on a Starlink-like constellation demonstrate that AC-MDQN achieves end-to-end latency within 1–2\% of constrained shortest-path solutions, remains robust under varying pointing jitter, and supports controllable hop–latency trade-offs through reward design. The results confirm that the proposed framework provides an efficient and physically consistent routing solution for large-scale optical LEO networks.
\end{abstract}


\maketitle


\section{Introduction}
\IEEEPARstart{L}{arge-scale} low-Earth-orbit (LEO) constellations have emerged as a key enabler for global broadband connectivity, offering low-latency, high-throughput links to underserved regions and supporting delay-sensitive services such as cloud gaming, financial trading, and real-time sensing \cite{salim2024cybersecurity}. 
To meet these increasing traffic demands, recent LEO systems are shifting from RF to optical inter-satellite links (ISLs). Unlike terrestrial FSO links that suffer from atmospheric absorption and scattering, laser beams propagate in the vacuum of space with minimal loss, enabling multi-Gb/s capacity, narrow beams, and compact terminals \cite{salim2024cybersecurity}. These advantages have motivated major constellations, such as recent generations of Starlink and OneWeb, to adopt optical crosslinks as their primary high-speed backbone \cite{wu2025enhancing}.

Routing in such networks is, however, fundamentally challenging. The LEO topology is highly dynamic, with thousands of satellites moving on multiple orbital planes and constantly changing visibility to each other and to ground gateways. Optical ISLs use tightly focused beams that are inherently sensitive to platform vibration, residual tracking errors, and terminal misalignment, which translate into non-negligible pointing losses and distance-dependent outage probabilities \cite{10681506, 10553228}. Moreover, intra-plane and inter-plane ISLs exhibit markedly different stability and tracking accuracy, leading to heterogeneous reliability across the constellation. As a result, a routing decision must jointly account for (i) geometric aspects such as link distance and great-circle structure, (ii) physical-layer QoS constraints such as outage probability and SNR thresholds, and (iii) network-layer metrics including end-to-end latency, hop count, and capacity constraints. The interaction of these factors, combined with the large constellation size, quickly renders classical routing approaches inadequate.

\subsection{Related Work and Research Gaps}
Routing in LEO and mega-constellation networks has been extensively investigated from algorithmic and architectural perspectives. Survey and modeling studies \cite{cao2022dynamic,Ghanbari2026NarrowBeams} emphasize the need for routing strategies that remain robust under rapidly evolving satellite positions and frequent topology changes. Various works propose deterministic or heuristic routing schemes tailored to LEO networks, including distributed on-demand routing \cite{stock2022distributed}, area-segmentation and grid-based approaches \cite{li2023leo,wang2023orbit}, complexity-reduced near-optimal routing such as MLDR \cite{wang2024mldr}, and recent link-identified architectures \cite{zhang2024link}. RF and optical connectivity have also been compared in terms of topology and delay trends \cite{de2025rf}. These efforts successfully capture geometric and temporal dynamics but treat inter-satellite links as abstract edges with fixed reachability or cost, without embedding physical-layer optical constraints such as beam divergence, pointing jitter, or class-dependent ISL reliability.

A complementary line of work applies deep reinforcement learning (DRL) to routing in dynamic non-terrestrial networks. DRL-based schemes in \cite{dong2023drl,chen2025deep} demonstrate the ability of learned policies to adapt to changing connectivity, balance load, and satisfy service-related constraints. Despite these advantages, the underlying link models remain simplified, relying on static capacities, queueing delay, or heuristic metrics. Physical feasibility constraints, such as outage probability, jitter-induced misalignment, and analytically derived distance limits, are not incorporated during action selection. Consequently, the resulting policies are not physics-informed and may choose routes that are infeasible for realistic optical ISLs.

Recent works have begun examining laser-based ISLs more explicitly. On-demand routing under dynamic laser connectivity is explored in \cite{bhattacharjee2024demand}, while \cite{yang2024analysis} studies laser ISL topology construction, and \cite{10437539} analyzes delay performance in dense constellations with ISLs. These studies highlight the importance of optical ISLs for low-latency global transport. However, optical links are typically modeled as ideal or fixed-reach once activated, and pointing jitter, divergence constraints, and outage-based reliability limitations are not converted into closed-form feasibility regions. The differences between intra-plane and inter-plane links are captured geometrically but not through their distinct tracking accuracy or jitter profiles.

Taken together, the above studies address dynamic topology, DRL-based routing, and laser ISL deployment, but none jointly integrates (i) a Gaussian-beam optical model with pointing jitter, (ii) analytic outage-derived maximum link ranges distinguishing intra- and inter-plane ISLs, and (iii) a scalable DRL framework that enforces these physical constraints through action masking and exploits geometric corridor structure to reduce decision complexity. This gap motivates the present work, which develops a physics-informed and geometry-aware routing framework that embeds optical feasibility directly into both the network model and the learning agent, enabling physically consistent and scalable routing for next-generation optical LEO constellations.

\subsection{Contributions}
This work addresses several key gaps in existing LEO routing research, including the absence of a physics-based optical ISL model in routing decisions, the lack of jitter-dependent intra-/inter-plane link separation, the absence of closed-form outage-derived range constraints, and the lack of scalable DRL routing mechanisms that enforce physical feasibility and exploit geometric corridor structure. The main contributions are summarized as follows:
\begin{itemize}
	\item We develop a physics-consistent routing model that incorporates a Gaussian-beam optical ISL formulation with pointing jitter, and we derive closed-form outage-based feasibility limits that yield class-dependent maximum link ranges for intra-plane and inter-plane connections.
	
	\item We reformulate latency–reliability–capacity constrained routing by embedding these optical feasibility constraints directly into the network graph, creating a geometry-aware and QoS-aware routing problem that cannot be captured by existing abstract link models.
	
	\item We design a scalable deep reinforcement learning framework in which optical feasibility is enforced through action masking, and the geometric concentration of near-optimal routes is exploited via an angular-corridor constraint, significantly reducing the effective decision space.
	
	\item We develop a latency- and reliability-aligned reward structure that enables the DRL agent to balance propagation delay, hop count, and outage risk in a dynamic LEO environment.
	
	\item We evaluate the effect of tracking accuracy on routing feasibility under realistic jitter levels and show that the proposed AC-MDQN remains reliable across a wide range of pointing-error conditions, with performance degrading gracefully as the optical constraints tighten.
	
	\item We investigate the impact of reward-parameter tuning and show that the agent can balance end-to-end latency against the number of relays, enabling explicit control over the degree to which satellites are involved along the route and providing flexibility under different load and QoS regimes.
\end{itemize}

\section{System Model}

\subsection{Network Entities and Graph Representation}
We consider a dynamic LEO constellation consisting of a set of satellites $\mathcal{S}$ uniformly distributed over multiple orbital planes, and a set of ground gateways $\mathcal{G}=\{s_g,d_g\}$. At time $t$, the network is modeled as a directed graph
\begin{align}
	\mathcal{V}(t) &= \mathcal{S}(t) \cup \mathcal{G}, \qquad
	\mathcal{E}(t) = \mathcal{E}_{\text{intra}}(t) \cup \mathcal{E}_{\text{inter}}(t) \cup \mathcal{E}_{\text{gw}}(t),
	\label{eq:graph_model}
\end{align}
where $\mathcal{E}_{\text{intra}}$ denotes intra-plane optical inter-satellite links (ISLs), $\mathcal{E}_{\text{inter}}$ denotes inter-plane ISLs, and $\mathcal{E}_{\text{gw}}$ includes satellite–gateway access and egress links.

\subsubsection{LEO Constellation Geometry}
The satellites are assumed to follow circular orbits at altitude $h$ with Earth radius $R_e$ and orbital radius $r = R_e + h$. Each orbital plane has inclination $\theta_{\mathrm{inc}}$ and right ascension of ascending node (RAAN) separated by $\Delta\Omega$. The angular position of satellite $k$ in plane $p$ at time $t$ is expressed as
\begin{align}
	\boldsymbol{r}_{p,k}(t) = r \, \boldsymbol{R}_z(\Omega_p) \, \boldsymbol{R}_x(\theta_{\mathrm{inc}})
	\begin{bmatrix}
		\cos(\omega t + \phi_{p,k}) \\ \sin(\omega t + \phi_{p,k}) \\ 0
	\end{bmatrix},
	\label{eq:sat_pos}
\end{align}
where $\omega = \sqrt{\mu / r^3}$ is the orbital angular velocity, $\mu$ is Earth’s gravitational constant, and $\phi_{p,k}$ is the phase offset within the plane. The relative motion of satellites induces time-varying link distances $\ell_e(t)$ and network topology $\mathcal{E}(t)$.

\subsubsection{Ground Gateways and Visibility}
Each gateway $g\in\mathcal{G}$ is located at a fixed geodetic coordinate $(\varphi_g,\lambda_g)$, mapped to ECEF position $\boldsymbol{r}_g = R_e [\cos\varphi_g\cos\lambda_g, \cos\varphi_g\sin\lambda_g, \sin\varphi_g]^T$.  
A satellite $i$ is visible to gateway $g$ only if its elevation angle $\varepsilon_{i,g}(t)$ satisfies
\begin{align}
	\varepsilon_{i,g}(t) \ge \varepsilon_{\min},
	\label{eq:elevation_mask}
\end{align}
where $\varepsilon_{\min}$ is the minimum elevation mask (typically $10^\circ$–$20^\circ$) accounting for atmosphere and line-of-sight blockage.

\subsection{Optical Inter-Satellite Link Model}
Each optical ISL operates with a Gaussian beam between nodes $i$ and $j$, characterized by transmit power $P_t$, aperture radii $(a_t,a_r)$, optical efficiencies $(\eta_t,\eta_r)$, and beam divergence angle $\theta_{\text{div}}$ that determines the beam expansion over distance $\ell_e(t)$.

\subsubsection{Gaussian Beam Propagation and Divergence}
The beam radius at range $\ell_e$ is given by \cite{Ghanbari2024UAVHAP}
\begin{align}
	w(\ell_e) = w_0 \sqrt{1+\left(\frac{\ell_e\,\theta_{\text{div}}}{w_0}\right)^{\!2}}
	\approx w_0 \sqrt{1+\left(\frac{\lambda \ell_e}{\pi w_0^2}\right)^{\!2}},
	\label{eq:beam_radius}
\end{align}
where $w_0$ is the beam waist radius at the transmitter and $\lambda$ is the optical wavelength. A smaller divergence angle $\theta_{\text{div}}$ (tighter beam) decreases diffraction loss but increases sensitivity to pointing errors; hence, $\theta_{\text{div}}$ influences the outage probability in a distance-dependent manner.

Under ideal alignment, the received optical power is
\begin{align}
	P_{r,0}(e,t) = P_t\,\eta_t\,\eta_r
	\left(\frac{a_r}{w(\ell_e)}\right)^{\!2}
	\exp(-\alpha\,\ell_e(t)),
	\label{eq:Pr0}
\end{align}
where $\alpha$ is the atmospheric attenuation coefficient, negligible for space-to-space ISLs \cite{Dabiri2025OAMFSO}.

\subsubsection{Pointing Error and Outage Probability}
Let the pointing error $\theta$ follow a zero-mean Gaussian distribution, $\theta \sim \mathcal{N}(0,\sigma_\theta^2)$, where $\sigma_\theta$ represents the residual tracking jitter caused by platform vibration and control dynamics.  
The effective lateral offset at range $\ell_e$ is $r_\perp = \ell_e \theta$, yielding the instantaneous pointing loss \cite{dabiri2019optimal}
\begin{align}
	L_{\text{pt}}(\theta,\ell_e) = \exp\!\Big(-2\frac{\ell_e^2\theta^2}{w^2(\ell_e)}\Big).
	\label{eq:pointing_loss}
\end{align}
Accordingly, the instantaneous received power is
\begin{align}
	P_r(e,t) = P_{r,0}(e,t)\,L_{\text{pt}}(\theta,\ell_e).
	\label{eq:Pr_final}
\end{align}
The corresponding instantaneous signal-to-noise ratio (SNR) is
\begin{align}
	\gamma_e(t) = \frac{P_r(e,t)}{N_0 B},
	\label{eq:snr_def}
\end{align}
where $N_0$ and $B$ are the noise spectral density and link bandwidth, respectively.

For a target SNR threshold $\gamma_{\text{th}}$, the outage probability is obtained by averaging over the pointing error:
\begin{align}
	P_{\text{out},e}(t)
	&= \Pr\big\{\gamma_e(t)<\gamma_{\text{th}}\big\} \nonumber\\[3pt]
	&= 2Q\!\left(
	\frac{1}{\sigma_\theta}
	\sqrt{
		\frac{w^2(\ell_e)}{2\,\ell_e^2}
		\ln\!\Big(\frac{P_{r,0}(e,t)}{\gamma_{\text{th}}N_0B}\Big)
	}
	\right),
	\label{eq:pout}
\end{align}
where $Q(\cdot)$ denotes the Gaussian tail function.  
A feasible ISL must satisfy the target reliability constraint
\begin{align}
	P_{\text{out},e}(t) \le P_{\text{out,th}},
	\label{eq:pout_thresh}
\end{align}
which implicitly relates allowable distance, divergence angle, and tracking accuracy.

\subsubsection{Intra-Plane vs. Inter-Plane Link Characteristics}
Intra-plane and inter-plane ISLs exhibit inherently different stability and tracking accuracy due to constellation geometry and relative satellite motion.  
Satellites within the same orbital plane maintain nearly constant angular separation and relative velocity, resulting in smaller pointing jitter variance $\sigma_{\theta,\text{intra}}^2$ and consequently lower outage probability.  
In contrast, inter-plane ISLs experience faster relative motion and more frequent alignment updates, leading to larger tracking error variance $\sigma_{\theta,\text{inter}}^2>\sigma_{\theta,\text{intra}}^2$ \cite{dabiri2025all}.  
Therefore, inter-plane links generally have reduced reliability for identical optical parameters, a factor explicitly captured through the outage probability model in \eqref{eq:pout}–\eqref{eq:pout_thresh} and considered in the subsequent routing formulation.

\subsection{Propagation Delay Model}
The pure propagation delay on link $e$ is determined by the geometric path length:
\begin{align}
	\tau_e(t) = \frac{\ell_e(t)}{c},
	\label{eq:prop_delay}
\end{align}
where $c$ is the speed of light in vacuum. The end-to-end propagation latency along a path $\mathcal{P}$ is therefore
\begin{align}
	T_{\text{prop}}(t) = \sum_{e\in\mathcal{P}} \frac{\ell_e(t)}{c}.
	\label{eq:e2e_prop_delay}
\end{align}
Additional queuing or processing delays will be introduced later in the routing formulation, while the physical model above captures the geometric latency component governed by constellation dynamics and beam alignment.

\section{Problem Formulation}
We seek a minimum-latency route from $s_g$ to $d_g$ on the network snapshot at time $t$. The graph and link sets follow \eqref{eq:graph_model}. Each feasible link must meet both the geometric visibility condition \eqref{eq:elevation_mask} and the reliability constraint \eqref{eq:pout_thresh}. The propagation delay of each link is given by \eqref{eq:prop_delay}.

\subsection{Feasible Link Set}
Routing decisions are restricted to links that are both visible and sufficiently reliable according to the outage probability model in \eqref{eq:pout}–\eqref{eq:pout_thresh}. Accordingly, we define
\begin{align}
	\mathcal{E}_{\mathrm{feas}}(t)
	&= 
	\Big\{\, e\in\mathcal{E}(t):~
	P_{\text{out},e}(t)\le P_{\text{out,th}}, \nonumber\\
	&\qquad\qquad
	\text{and if } e\in\mathcal{E}_{\text{gw}}(t),~
	\varepsilon_{i,g}(t)\ge \varepsilon_{\min}\,\Big\}.
	\label{eq:feasible_set_pf}
\end{align}
This feasible set implicitly reflects the effects of beam divergence, link distance, and tracking jitter on link reliability through \eqref{eq:pout}.

Let $x_e\in\{0,1\}$ denote whether link $e\in\mathcal{E}_{\mathrm{feas}}(t)$ is selected in the routing path.  
The binary flow-conservation rule enforces a single end-to-end path between source and destination gateways:
\begin{align}
	\sum_{e=(v\to j)} x_e - \sum_{e=(i\to v)} x_e &=
	\begin{cases}
		1, & v=s_g,\\
		-1, & v=d_g,\\
		0, & \text{otherwise},
	\end{cases}
	\quad \forall v\in\mathcal{V}(t).
	\label{eq:flow_cons}
\end{align}
The hop count along the path is given by
\begin{align}
	H \;=\; \sum_{e\in\mathcal{E}_{\mathrm{feas}}(t)} x_e.
	\label{eq:hop_count}
\end{align}

\subsection{End-to-End Latency Model}
The total latency includes propagation delay and additional per-hop overheads such as switching or queuing. Using \eqref{eq:prop_delay}, we write
\begin{align}
	T(t)
	&= \sum_{e\in\mathcal{E}_{\mathrm{feas}}(t)} \frac{\ell_e(t)}{c}\,x_e 
	\;+\; \tau^{\text{hop}} H 
	\;+\; \sum_{e\in\mathcal{E}_{\mathrm{feas}}(t)} \tau^{\text{queue}}_e\,x_e,
	\label{eq:e2e_latency}
\end{align}
where $\tau^{\text{hop}}$ is the average processing delay per relay and $\tau^{\text{queue}}_e$ represents queuing delay, which may depend on traffic load or buffer dynamics.  
In light-load conditions, $\tau^{\text{queue}}_e \le \bar\tau^{\text{queue}}$, and a conservative bound yields
\begin{align}
	T(t)\;\le\; \frac{1}{c}\sum_{e} \ell_e(t)\,x_e 
	\;+\; \big(\tau^{\text{hop}}+\bar\tau^{\text{queue}}\big)H.
	\label{eq:e2e_latency_bound}
\end{align}

Each link’s outage probability $P_{\text{out},e}(t)$ follows \eqref{eq:pout}, and the end-to-end reliability requirement $R_{\min}\in(0,1)$ is enforced as
\begin{align}
	\sum_{e\in\mathcal{E}_{\mathrm{feas}}(t)} \log\!\big(1-P_{\text{out},e}(t)\big)\,x_e 
	\;\ge\; \log R_{\min}.
	\label{eq:e2e_reliability}
\end{align}
For a flow demand $r$ and supported rate $C_e$ on each selected link, we impose
\begin{align}
	\sum_{e\in\mathcal{E}_{\mathrm{feas}}(t)} \frac{r}{C_e}\,x_e \;\le\; 1,
	\label{eq:e2e_capacity}
\end{align}
ensuring that the aggregated load remains within the available link capacities.  
If a global latency budget $D_{\max}$ is specified, we additionally require
\begin{align}
	T(t)\;\le\; D_{\max}.
	\label{eq:e2e_delay_budget}
\end{align}

\subsection{Objectives}
The routing objective is to minimize the end-to-end latency while maintaining the reliability and capacity constraints.

\subsubsection{Exact Latency Minimization}
\begin{align}
	\min_{\{x_e\}} \quad & T(t)
	\label{eq:obj_exact}\\
	\text{s.t.}\quad &
	\eqref{eq:feasible_set_pf},\ 
	\eqref{eq:flow_cons},\ 
	\eqref{eq:e2e_reliability},\ 
	\eqref{eq:e2e_capacity},\ 
	\eqref{eq:e2e_delay_budget},\ 
	x_e\in\{0,1\}.
	\notag
\end{align}

\subsubsection{Length–Hop Surrogate Formulation}
By substituting the upper bound \eqref{eq:e2e_latency_bound}, define $\beta=\tau^{\text{hop}}+\bar\tau^{\text{queue}}$, and minimize
\begin{align}
	\min_{\{x_e\}} \quad &
	\frac{1}{c}\sum_{e\in\mathcal{E}_{\mathrm{feas}}(t)} \ell_e(t)\,x_e 
	\;+\; \beta \sum_{e\in\mathcal{E}_{\mathrm{feas}}(t)} x_e
	\label{eq:obj_surrogate}\\
	\text{s.t.}\quad &
	\eqref{eq:feasible_set_pf},\ 
	\eqref{eq:flow_cons},\ 
	\eqref{eq:e2e_reliability},\ 
	\eqref{eq:e2e_capacity},\ 
	x_e\in\{0,1\}.
	\notag
\end{align}
The first term captures the physical propagation component of latency determined by link geometry, while the second penalizes excessive relays, encouraging routes with fewer hops whenever link reliability allows longer connections through wider beam divergence or improved tracking.

Programs \eqref{eq:obj_exact} and \eqref{eq:obj_surrogate} are combinatorial due to the binary routing variables $x_e$.  
The reliability constraint \eqref{eq:e2e_reliability} couples the geometric link lengths and optical parameters (beam divergence and pointing jitter) through the outage model \eqref{eq:pout}.  
Moreover, the network topology and link distances evolve over time as satellites move according to \eqref{eq:graph_model}.  
These properties motivate a learning-based route selection method to approximate near-optimal decisions under dynamic geometry and probabilistic link quality.

\section{Analytical Foundations and Feasibility Constraints}

\subsection{Computational Complexity}
We first characterize the computational complexity of the routing problem with reliability and capacity constraints as formulated in Section~IV.

\begin{proposition}
	The optimization problem \eqref{eq:obj_exact} (and its surrogate form \eqref{eq:obj_surrogate}), 
	subject to the outage-based feasibility constraint \eqref{eq:feasible_set_pf} 
	and the end-to-end reliability constraint \eqref{eq:e2e_reliability}, 
	is NP-hard.
\end{proposition}

\textit{Sketch of Proof:}
	Consider a simplified instance of \eqref{eq:obj_exact} where queuing delays are neglected 
	and link reliabilities $P_{\text{out},e}(t)$ are treated as fixed constants. 
	Let $\mathcal{P}$ denote a candidate path from $s_g$ to $d_g$. 
	The end-to-end reliability can be expressed as
	\begin{align}
		R(\mathcal{P}) = \prod_{e\in\mathcal{P}} (1-P_{\text{out},e}),
		\label{eq:proof_rel}
	\end{align}
	while the total latency cost is additive:
	\begin{align}
		C(\mathcal{P}) = \sum_{e\in\mathcal{P}} \Big(\frac{\ell_e(t)}{c} + \beta\Big).
		\label{eq:proof_cost}
	\end{align}
	The problem then reduces to finding a path $\mathcal{P}$ that minimizes $C(\mathcal{P})$ 
	subject to a minimum reliability requirement $R(\mathcal{P}) \ge R_{\min}$. 
	Taking the logarithm of \eqref{eq:proof_rel} converts the multiplicative constraint into an additive form:
	\begin{align}
		\sum_{e\in\mathcal{P}} \log(1-P_{\text{out},e}) \ge \log R_{\min},
		\label{eq:proof_log}
	\end{align}
	which is equivalent to the additive reliability expression in \eqref{eq:e2e_reliability}. 
	Therefore, the simplified version of \eqref{eq:obj_exact} 
	becomes a constrained shortest-path problem with an additive budget constraint \eqref{eq:proof_log}, 
	a class of problems known to be NP-hard \cite{kuipers2005conditions}. 
	Since the full formulation further couples $P_{\text{out},e}(t)$ with geometric and optical parameters 
	such as $\ell_e(t)$ and $\theta_{\text{div},e}$, it strictly generalizes this NP-hard core. 
	Consequently, the overall routing problem remains NP-hard. 
	Note that in the full dynamic setting, link reliabilities and delays evolve over time due to UAV mobility 
	and geometry variations. This temporal coupling further increases problem complexity 
	but does not alter the NP-hard nature established above.

\subsection{Temporal Dynamics and Snapshot Optimization}
The network is time-varying through satellite motion, hence $\ell_e(t)$ and $P_{\text{out},e}(t)$ evolve with $t$ (cf. \eqref{eq:graph_model}, \eqref{eq:e2e_prop_delay}, and \eqref{eq:pout}). We adopt a \emph{snapshot} optimization at decision epochs $t_k$, producing a route (and, below, per-link optical settings) for interval $[t_k,t_{k+1})$. This yields a sequential decision process over time; in next section we model the sequence $\{t_k\}$ as an MDP to learn routing policies that adapt to geometry and link statistics \cite{barto2021reinforcement}.

\subsection{Outage Constraint as a Function of Distance}
From \eqref{eq:beam_radius} and \eqref{eq:pout}, the outage probability on link $e$ at time $t$ depends explicitly on the geometric distance $\ell_e(t)$ and the divergence angle $\theta_{\text{div}}$ through the beam radius $w(\ell_e)$. 
For inter-satellite ranges $\ell_e \gg \pi w_0^2/\lambda$, we use the far-field approximation $w(\ell_e)\!\approx\!\theta_{\text{div}}\ell_e$, which is accurate for LEO ISLs. 
By rewriting \eqref{eq:Pr0} under boresight coupling with a fixed receive aperture, the on-axis received power can be expressed as 
\[
P_{r,0}(e,t;\theta_{\text{div}})=\kappa_e(t)\,w^{-2},
\]
which leads to the following closed-form expression for the outage probability:
\begin{align}
	P_{\text{out},e}\big(t;\theta_{\text{div}}\big)
	= 2Q\!\left(
	\frac{\theta_{\text{div}}}{\sqrt{2}\,\sigma_{\theta,e}}\,
	\sqrt{\,a_e(t)-2\ln(\theta_{\text{div}}\ell_e)\,}
	\right),
	\label{eq:pout_theta_sigma}
\end{align}
where $a_e(t)\triangleq \ln\!\big(\kappa_e(t)/(\gamma_{\text{th}}N_0B)\big)$ collects the range-independent parameters, and $\sigma_{\theta,e}$ denotes the RMS pointing jitter.

\begin{proposition}[Maximum feasible link length as a function of jitter]
	Fix a link $e$ and time $t$ with parameters $(a_e(t),\sigma_{\theta,e},\theta_{\max})$.
	For a given outage target $P_{\text{out,th}}\in(0,1)$, define $q_{\text{th}}=Q^{-1}(P_{\text{out,th}}/2)>0$. 
	Under the far-field model, the maximum feasible link length $\ell_{\max}(t;\sigma_{\theta,e})$ satisfying $P_{\text{out},e}\!\big(t;\theta_{\text{div}}\big)\le P_{\text{out,th}}$ for some admissible $\theta_{\text{div}}\!\in(0,\theta_{\max}]$ is
	\begin{align}
		\ell_{\max}(t;\sigma_{\theta,e})
		= \frac{\exp\!\big(\frac{a_e(t)}{2}-\tfrac{1}{2}\big)}{\sqrt{2}\,\sigma_{\theta,e}\,q_{\text{th}}},
		\label{eq:lmax_unconstrained_final}
	\end{align}
	provided that the corresponding optimal divergence
	\begin{align}
		\theta_{\text{div}}^\star(\ell_{\max},t)
		= \frac{1}{\ell_{\max}}\,
		\exp\!\Big(\frac{a_e(t)}{2}-\tfrac{1}{2}\Big)
		\label{eq:theta_star_for_lmax}
	\end{align}
	does not exceed $\theta_{\max}$. 
	If $\theta_{\text{div}}^\star>\theta_{\max}$, the range is limited by the divergence cap:
	\begin{align}
		\ell_{\max}^{(\mathrm{bd})}(t;\sigma_{\theta,e},\theta_{\max})
		= \frac{1}{\theta_{\max}}\,
		\exp\!\left(
		\frac{a_e(t)}{2}
		- \frac{(\sqrt{2}\sigma_{\theta,e}q_{\text{th}})^2}{2\,\theta_{\max}^2}
		\right).
		\label{eq:lmax_bd_final}
	\end{align}
\end{proposition}

\textit{Proof:}
	The detailed proof is provided in Appendix~\ref{app:proof_lmax}.

Equation \eqref{eq:lmax_unconstrained_final} clearly shows that $\ell_{\max}\!\propto\!1/\sigma_{\theta,e}$: higher jitter directly reduces the maximum reliable range even under optimal beam tuning. 

\begin{table}[t]
	\centering
	\caption{Simulation parameters used for Proposition evaluation and exhaustive 2D verification.}
	\label{tab:sim_params_full}
	\begin{tabular}{l c}
		\hline
		Parameter & Value \\
		\hline
		Optical wavelength $\lambda$ & $1550$ nm \\
		Transmit power $P_t$ & $1$ W \\
		Transmitter/receiver efficiencies $(\eta_t,\eta_r)$ & $(0.5,\,0.5)$ \\
		Receiver aperture radius $a_r$ & $5$ cm \\
		System loss $L_{\text{const}}$ & $10$ dB \\
		Noise$\times$bandwidth $N_0B$ & $10^{-12}$ W \\
		SNR threshold $\gamma_{\text{th}}$ & $10$ (linear) \\
		Maximum divergence $\theta_{\max}$ & $1$ mrad \\
		Outage threshold $P_{\text{out,th}}$ & $10^{-4}$ \\
		$Q^{-1}(P_{\text{out,th}})$ & $q_{\text{th}}{=}3.719016$ \\
		$a_e(t)$ (from $\ln(\kappa_e/(\gamma_{\text{th}}N_0B))$) & $15.648092$ \\
		\hline
	\end{tabular}
\end{table}

\begin{table}[t]
	\centering
	\caption{Maximum feasible ISL range vs.\ pointing jitter: analytical Proposition (\eqref{eq:lmax_unconstrained_final}) versus exhaustive 2D optimization.}
	\label{tab:lmax_prop_vs_exhaustive}
	\begin{tabular}{c c c}
		\hline
		$\sigma_{\theta,e}$ [$\mu$rad] & $\ell_{\max}$ (Proposition) [km] & $\ell_{\max}$ (Exhaustive) [km] \\
		\hline
		100 & 2883.0 & 2857.1 \\
		200 & 1441.5 & 1458.8 \\
		300 & 961.0  & 951.4  \\
		400 & 718.9  & 723.9  \\
		500 & 526.6  & 519.9  \\
		\hline
	\end{tabular}
\end{table}

Table~\ref{tab:sim_params_full} lists the numerical parameters adopted in the simulation and closed-form evaluation. 
Table~\ref{tab:lmax_prop_vs_exhaustive} then compares the maximum feasible inter-satellite link (ISL) range predicted by the closed-form Proposition (cf.\ \eqref{eq:lmax_unconstrained_final}) against an exhaustive two-dimensional search over $(\ell,\theta_{\text{div}})$ using the same parameter set. 
The tight agreement (within $0.5$–$1.5\%$) validates the analytical result under the adopted far-field model and link budget. 
Equation~\eqref{eq:lmax_unconstrained_final} clearly shows that $\ell_{\max}\!\propto\!1/\sigma_{\theta,e}$: higher pointing jitter directly reduces the maximum reliable range even under optimal divergence tuning. 
For a fixed link geometry, this expression provides an analytical upper bound on the achievable inter-satellite distance for given optical power, divergence limit, and tracking accuracy.

\subsection{Reformulated Optimization with Link-Class-Specific Constraints}
The analytical relations in \eqref{eq:lmax_unconstrained_final}–\eqref{eq:lmax_bd_final} provide a closed-form characterization of the maximum feasible optical link range as a function of the pointing jitter $\sigma_{\theta,e}$ and divergence tuning limits $(\theta_{\min},\theta_{\max})$.  
Since the divergence angle is autonomously adjusted by the terminal optics to meet the outage target, it no longer appears as an optimization variable.  
This elimination of $\theta_{\text{div}}$ notably reduces the dimensionality of the routing problem, while keeping the physical feasibility inherently enforced through range constraints derived from the analytical model.
These results allow us to refine the feasible set $\mathcal{E}_{\mathrm{feas}}(t)$ by introducing explicit link-class-dependent geometric bounds that distinguish intra-plane and inter-plane inter-satellite links.

\subsubsection{Class-Specific Range Constraints}
Using \eqref{eq:lmax_unconstrained_final}, the maximum feasible link lengths for each class are expressed as
\begin{align}
	L_{\max}^{\text{intra}}(t)
	&= \frac{\exp\!\big(\tfrac{a_e(t)}{2}-\tfrac{1}{2}\big)}{\sqrt{2}\,\sigma_{\theta,\text{intra}}\,q_{\text{th}}}, \label{eq:lmax_intra_final}\\[2pt]
	L_{\max}^{\text{inter}}(t)
	&= \frac{\exp\!\big(\tfrac{a_e(t)}{2}-\tfrac{1}{2}\big)}{\sqrt{2}\,\sigma_{\theta,\text{inter}}\,q_{\text{th}}}, 
	\qquad \sigma_{\theta,\text{inter}}>\sigma_{\theta,\text{intra}}. \label{eq:lmax_inter_final}
\end{align}
These bounds depend on the instantaneous geometry via $a_e(t)$ and explicitly capture the degradation of tracking accuracy in inter-plane links.  
Since $\sigma_{\theta,\text{inter}}>\sigma_{\theta,\text{intra}}$, it follows that $L_{\max}^{\text{inter}}(t)<L_{\max}^{\text{intra}}(t)$, implying stricter feasibility for cross-plane links.

Accordingly, the link-length constraints are imposed as
\begin{align}
	\ell_e(t) \le 
	\begin{cases}
		L_{\max}^{\text{intra}}(t), & e\in\mathcal{E}_{\text{intra}}(t),\\[3pt]
		L_{\max}^{\text{inter}}(t), & e\in\mathcal{E}_{\text{inter}}(t),
	\end{cases}
	\label{eq:link_length_class}
\end{align}
which incorporate the reliability–geometry coupling directly through the outage-based range limits.

\subsubsection{Unified Optimization with Link-Class Distinction}
By combining \eqref{eq:e2e_latency_bound}, \eqref{eq:e2e_reliability}, and \eqref{eq:link_length_class}, the latency–reliability–geometry-aware routing problem can be reformulated as
\begin{align}
	\min_{\{x_e\}} \quad &
	\frac{1}{c}\sum_{e\in\mathcal{E}_{\mathrm{feas}}(t)} \ell_e(t)\,x_e 
	+ \beta \sum_{e\in\mathcal{E}_{\mathrm{feas}}(t)} x_e
	\label{eq:obj_final}\\
	\text{s.t.}\quad &
	\eqref{eq:flow_cons},~\eqref{eq:e2e_capacity},~\eqref{eq:e2e_reliability}, \nonumber\\
	& \ell_e(t)\le L_{\max}^{\text{intra}}(t),~ e\in\mathcal{E}_{\text{intra}}(t), \nonumber\\
	& \ell_e(t)\le L_{\max}^{\text{inter}}(t),~ e\in\mathcal{E}_{\text{inter}}(t), \nonumber\\
	& x_e\in\{0,1\},\ \forall e\in\mathcal{E}_{\mathrm{feas}}(t). \nonumber
\end{align}
The first summation term captures geometric propagation delay, and the second penalizes the hop count $H$.  
Constraints \eqref{eq:link_length_class} ensure that each link type satisfies its respective range limit, effectively embedding the optical physics into the routing layer.

The reformulated problem \eqref{eq:obj_final} preserves the combinatorial structure of the original optimization but eliminates the continuous divergence variable, yielding a lower-dimensional and more tractable problem.  
It remains fully consistent with the analytical feasibility conditions derived from the optical model, ensuring that every admissible route adheres to the physical constraints imposed by beam divergence, jitter, and geometry.  
The resulting intra/inter-plane heterogeneity provides a realistic and model-aligned foundation for the DRL-based framework developed in the next section, where these class-dependent link constraints are embedded directly into the learning environment and reward design.

\section{Proposed Learning-Based Routing Framework}

\subsection{Scientific Motivation and Design Principles}
Conventional constrained shortest-path algorithms, such as extended Dijkstra, label-setting, and dynamic programming, are not scalable for dynamic LEO topologies where thousands of moving nodes and time-varying link constraints make precomputation intractable \cite{cao2022dynamic,dong2023drl}.  
In addition, the feasible set of links $\mathcal{E}_{\mathrm{feas}}(t)$ changes with both geometry and reliability, as described in \eqref{eq:lmax_intra_final}–\eqref{eq:lmax_inter_final}.  
Hence, a deterministic optimizer must be re-executed at every snapshot, incurring exponential growth in complexity.  

To overcome this, we formulate routing as a sequential decision process and use DRL to learn a policy that maps local network states to routing actions.  
The agent learns from repeated snapshot interactions to minimize end-to-end latency under dynamic feasibility constraints.  
The use of DRL is scientifically justified because \cite{luong2019applications}:
\begin{itemize}
	\item It provides a data-driven approximation of the optimal mapping from topology and link conditions to routing decisions, avoiding the NP-hard combinatorial search in each snapshot.
	\item It naturally accommodates stochastic link availability and varying graph connectivity.
	\item The learning objective can be aligned directly with the latency surrogate in \eqref{eq:obj_final}.
\end{itemize}

\subsection{MDP Formulation}
At each decision epoch $t_k$, routing is modeled as an MDP:
\begin{align}
	\mathcal{M}(t_k)=\big(\mathcal{S},\mathcal{A},\mathcal{T},r,\gamma\big),
	\label{eq:mdp_tuple}
\end{align}
where the components are defined as follows.

\subsubsection{State Representation}
The agent observes a compact feature vector extracted from the local graph at node $v$:
\begin{align}
	s = [\, d_v,\; \delta_v,\; \rho_v,\; \upsilon_v, \; c_v\,],
	\label{eq:state_vec_final}
\end{align}
where  
$d_v=\|\boldsymbol{r}_v-\boldsymbol{r}_{d_g}\|$ is the Euclidean distance to the destination’s nearest satellite,  
$\delta_v$ is the normalized best one-step improvement in distance,  
$\rho_v$ is the ratio of feasible neighbors to the maximum possible degree,  
$\upsilon_v$ flags revisit events,  
and $c_v$ represents the normalized average link cost of the available actions.  
All features are normalized to $[0,1]$ to ensure stability and invariance across different constellation scales.

\subsubsection{Feasible Actions}
At node $v$, the feasible action set is
\begin{align}
	&\mathcal{A}(v)
	=
	\Big\{e=(v\!\to\! j)\in\mathcal{E}(t_k):
	\ell_e(t_k)\le L_{\max}^{\text{class}(e)}(t_k),~ \nonumber \\
	&\text{and if } e\in\mathcal{E}_{\text{gw}}(t_k),~\varepsilon_{i,g}(t_k)\ge \varepsilon_{\min}
	\Big\}.
	\label{eq:feasible_action_set_final}
\end{align}
If $\mathcal{A}(v)=\varnothing$, the episode terminates with failure.  
This deterministic transition is consistent with the static snapshot assumption at $t_k$.

\subsection{Reward Design: Latency-Driven vs. Reliability-Aware Schemes}
We derive learning signals directly from the routing objective to ensure policy optimality aligns with system-level goals. Starting from the surrogate latency objective in \eqref{eq:obj_final}, define the per-hop physical cost
\begin{align}
	\mathcal{C}(\mathcal{P};t_k)
	\triangleq
	\frac{1}{c}\sum_{e\in\mathcal{P}} \ell_e(t_k) \;+\; \beta\,H.
	\label{eq:path_cost_def}
\end{align}
Let a trajectory be $\mathcal{P}=\{e_0,\ldots,e_{T-1}\}$ with nodes $\{v_0,\ldots,v_T\}$, $v_0=s_g$, and $v_T=d_g$. The per-step hop cost decomposes as
\begin{align}
	\mathcal{C}(\mathcal{P};t_k)
	=
	\sum_{t=0}^{T-1}\left(\frac{\ell_{e_t}(t_k)}{c}+\beta\right).
	\label{eq:path_cost_decomp}
\end{align}
Hence a natural choice is to set instantaneous reward as the negative of the incremental cost so that maximizing return minimizes \eqref{eq:path_cost_def}. To accelerate credit assignment while preserving optimality, we use potential-based shaping with a scalar potential $\Phi:\mathcal{S}\to\mathbb{R}$ \cite{ng1999policy}.

\subsubsection{Latency-aligned shaping}
Let $d_v$ be the geometric distance feature at node $v$. Choose the potential
\begin{align}
	\Phi(s)\triangleq \eta\,d_v,
	\quad \eta\ge 0,
	\label{eq:potential_def}
\end{align}
and define the shaped instantaneous signal
\begin{align}
	r^{(L)}_t
	&\triangleq
	-\left(\frac{\ell_{e_t}(t_k)}{c}+\beta\right) \nonumber \\
	&~~~+\big(\Phi(s_t)-\Phi(s_{t+1})\big)
	+\zeta\,\mathbf{1}\{v_{t+1}=d_g\},
	\label{eq:reward_latency_deriv}
\end{align}
where $\zeta>0$ is a terminal bonus that helps with sparse goal signals \cite{ng1999policy}. The shaping term telescopes over the episode:
\begin{align}
	\sum_{t=0}^{T-1}\big(\Phi(s_t)-\Phi(s_{t+1})\big)
	= \Phi(s_0)-\Phi(s_T)
	= \eta\,[d_{s_0}-d_{d_g}],
	\label{eq:telescoping_identity}
\end{align}
and $d_{d_g}=0$ by construction. Therefore, the cumulative return satisfies
\begin{align}
	\sum_{t=0}^{T-1} r^{(L)}_t
	= -\mathcal{C}(\mathcal{P};t_k) + \eta\,d_{s_0} + \zeta,
	\label{eq:return_equals_cost}
\end{align}
which shows that maximizing return is equivalent to minimizing \eqref{eq:path_cost_def} up to trajectory-independent constants. This is the standard policy invariance result for potential-based shaping \cite{arulkumaran2017deep}. Unfolding \eqref{eq:reward_latency_deriv} yields the explicit per-step form used in \eqref{eq:reward_latency}:
\begin{align}
	r^{(L)}(s,e,s')
	=
	-\frac{\ell_e(t_k)}{c}
	-\beta
	+\eta\,(d_v-d_{j})
	+\zeta\,\mathbf{1}\{j=d_g\}.
	\label{eq:reward_latency}
\end{align}

\subsubsection{Incorporating reliability via Lagrangian relaxation}
The end-to-end reliability constraint \eqref{eq:e2e_reliability} can be written as
\begin{align}
	\sum_{e\in\mathcal{P}} \log\!\big(1-P_{\text{out},e}(t_k)\big)
	\;\ge\; \log R_{\min}.
	\label{eq:reliability_constraint_again}
\end{align}
In the high-reliability regime $P_{\text{out},e}\ll 1$, the first-order approximation $\log(1-p)\approx -p$ gives
\begin{align}
	\sum_{e\in\mathcal{P}} P_{\text{out},e}(t_k) \;\le\; B,
	\qquad B \triangleq -\log R_{\min},
	\label{eq:reliability_linearized}
\end{align}
which is a standard linearization used in optical-link design \cite{arulkumaran2017deep}. A Lagrangian relaxation of the constrained problem augments \eqref{eq:path_cost_def} with a penalty on the reliability budget:
\begin{align}
	\min_{\mathcal{P}}
	\left[
	\frac{1}{c}\sum_{e\in\mathcal{P}} \ell_e(t_k) + \beta\,H
	+ \lambda\!\left(\sum_{e\in\mathcal{P}} P_{\text{out},e}(t_k)-B\right)
	\right],
	\label{eq:lagrangian_relax}
\end{align}
where $\lambda\ge 0$ is the dual variable \cite{arulkumaran2017deep}. The corresponding per-step incremental cost becomes
\begin{align}
	\frac{\ell_e(t_k)}{c}+\beta+\lambda\,P_{\text{out},e}(t_k),
	\label{eq:per_step_cost_rel}
\end{align}
which suggests the reliability-aware shaped reward
\begin{align}
	r^{(R)}(s,e,s')
	&=
	-\frac{\ell_e(t_k)}{c}
	-\beta
	-\lambda\,P_{\text{out},e}(t_k) \nonumber \\ ~~~&
	+\eta\,(d_v-d_{j})
	+\zeta\,\mathbf{1}\{j=d_g\}.
	\label{eq:reward_reliability}
\end{align}
The same telescoping argument as in \eqref{eq:return_equals_cost} shows that maximizing $\sum_t r^{(R)}_t$ minimizes the Lagrangian in \eqref{eq:lagrangian_relax} up to constants, thus aligning learning with the constrained objective \cite{arulkumaran2017deep}.

\subsubsection{Choosing and calibrating the reward}
The appropriate design depends on operating conditions:
\begin{align}
	&\text{Latency-dominant: } \sigma_{\theta,e}^2 \text{ small and } P_{\text{out},e}\text{ negligible} \nonumber \\&~~~~~~~
	\Rightarrow r^{(L)} \text{ yields faster convergence}, \label{eq:regime_latency}\\
	&\text{Reliability-limited: } \sigma_{\theta,e}^2 \text{ large or } R_{\min}\uparrow \nonumber \\&~~~~~~~
	\Rightarrow r^{(R)} \text{ improves stability and feasibility}. \label{eq:regime_reliability}
\end{align}
To meet the reliability budget $B$ in expectation, the dual weight can be adapted online using a projected subgradient step \cite{arulkumaran2017deep}:
\begin{align}
	\lambda_{k+1}
	=
	\big[\lambda_k + \alpha_\lambda\,\big(S_k - B\big)\big]_+,
	\quad
	S_k \triangleq \mathbb{E}\!\left[\sum_{t} P_{\text{out},e_t}(t_k)\right],
	\label{eq:dual_update_rule}
\end{align}
with stepsize $\alpha_\lambda>0$ and projection $[\cdot]_+$ onto $\mathbb{R}_{\ge 0}$. This increases $\lambda$ if the realized outage budget is exceeded and decreases it otherwise, steering the policy toward the target reliability.

\begin{remark}
	For regimes where the first-order approximation in \eqref{eq:reliability_linearized} is loose, one may replace $P_{\text{out},e}$ in \eqref{eq:reward_reliability} with the convex surrogate $-\log(1-P_{\text{out},e})$ to better reflect the multiplicative nature of end-to-end success \cite{dong2024reachability}, at the cost of evaluating a logarithm per step.
\end{remark}

\subsection{Learning Algorithm: Masked DQN for Dynamic LEO Graphs}
The discrete and locally bounded nature of routing decisions in LEO networks makes Deep Q-Networks (DQN) an effective framework for learning optimal next-hop policies.  
At each relay node, the agent faces a small discrete action space limited to feasible outgoing links.  
Hence, instead of exploring the full network graph, the agent can focus on dynamically changing local neighborhoods.  

We employ a \emph{masked DQN} design to ensure physical consistency with link feasibility constraints.  
At every decision step, the infeasible links violating the geometric or reliability bounds in \eqref{eq:feasible_action_set_final} are masked from the output layer of the Q-network.  
This guarantees that the agent never selects an invalid action, even during early exploration, which significantly stabilizes learning \cite{luong2019applications}.  

The Q-network parameterized by $\theta$ estimates the state–action value function $Q_\theta(s,a)$ as the expected cumulative discounted return.  
For a transition $(s,a,r,s')$, the temporal-difference (TD) target and mean-squared loss are given by
\begin{align}
	y &= r + \gamma \max_{a'\in \mathcal{A}(v')} Q_{\bar\theta}(s',a'), \nonumber \\[3pt]
	\mathcal{L}(\theta) &= \mathbb{E}\big[(Q_\theta(s,a)-y)^2\big],
	\label{eq:dqn_update_final}
\end{align}
where $Q_{\bar\theta}$ denotes the target network, updated via Polyak averaging as
\begin{align}
	\bar\theta \leftarrow \tau\,\theta + (1-\tau)\,\bar\theta, \quad 0<\tau<1.
	\label{eq:polyak_final}
\end{align}
This approach provides a stable approximation to the Bellman optimality operator under the nonstationary geometry of LEO constellations.

\vspace{4pt}
\begin{algorithm}[t]
	\caption{Masked DQN Routing in Dynamic Optical LEO Networks}
	\label{alg:dqn_routing_updated}
	\begin{algorithmic}[1]
		\State Initialize Q-network $Q_\theta$ and target network $Q_{\bar\theta}$ with random weights
		\State Initialize replay buffer $\mathcal{D}$, set $\varepsilon \gets 1$, episode counter $k \gets 0$
		\While{training not converged}
		\State Sample constellation snapshot at time $t_k$ and construct $\mathcal{E}_{\mathrm{feas}}(t_k)$
		\State Set current node $v \gets s_g$, cumulative return $R \gets 0$
		\For{$t=1$ to $T_{\max}$}
		\State Derive state $s_t$ from \eqref{eq:state_vec_final}
		\State Enumerate feasible actions $\mathcal{A}(v)$ via \eqref{eq:feasible_action_set_final}
		\State Select action $a_t$ using $\varepsilon$-greedy policy with masking on $\mathcal{A}(v)$
		\State Execute $a_t$, move to node $j$, and compute reward:
		\begin{align*}
			r_t &\gets 
			\begin{cases}
				r^{(L)}(s_t,a_t,s_{t+1}), & \text{if latency-driven mode},\\[2pt]
				r^{(R)}(s_t,a_t,s_{t+1}), & \text{if reliability-aware mode.}
			\end{cases}
		\end{align*}
		\State Store $(s_t,a_t,r_t,s_{t+1})$ in buffer $\mathcal{D}$
		\State Sample mini-batch from $\mathcal{D}$ and update $Q_\theta$ via \eqref{eq:dqn_update_final}
		\State Update target network $\bar\theta \leftarrow \tau\theta+(1-\tau)\bar\theta$
		\State $R \gets R + r_t$, $v \gets j$
		\If{$v=d_g$ or $\mathcal{A}(v)=\emptyset$}
		\State \textbf{break}
		\EndIf
		\EndFor
		\State $\varepsilon \gets \max(\varepsilon_{\min},\,\varepsilon - \Delta\varepsilon)$
		\EndWhile
	\end{algorithmic}
\end{algorithm}


\section{Geometric Constraints and Learning Refinement}

\subsection{Complexity and Practical Limits}
The masked DQN reduces the combinatorial action space from the full network size $|\mathcal{E}(t_k)|$ 
to a bounded local degree $A_{\max}$ corresponding to the number of feasible neighbors per satellite.  
Each training episode of average length $T_{\mathrm{ep}}$ thus scales as
\begin{align}
	\text{cost per episode} = O\big(A_{\max}\,T_{\mathrm{ep}}\big),
	\label{eq:mdqn_cost}
\end{align}
which makes training feasible even in mega-constellations with $10^4$–$10^5$ satellites.

Despite this, unconstrained learning may face a vanishing success probability on long intercontinental routes due to sparse terminal reward and compounding exploration.  
Let $H$ denote the hop count of an optimal route between a source and a distant destination.  
Consider an $\varepsilon$–greedy policy with masking where at step $t$ the feasible branching factor is $A_t\!\le\!A_{\max}$.  
Define the per-step success probability lower bound
\begin{align}
	p_t \;\ge\; \frac{\varepsilon}{A_t},
	\label{eq:pt_lower}
\end{align}
which holds when the exploratory move selects the unique optimal neighbor uniformly among $A_t$ feasible actions.  
A conservative lower bound on the probability of sampling a fully successful trajectory in one episode is then
\begin{align}
	p_{\text{succ}} \;\ge\; \prod_{t=1}^{H} p_t \;\ge\; \Big(\tfrac{\varepsilon}{\bar A}\Big)^{H},
	\qquad
	\bar A \triangleq \max_{t} A_t.
	\label{eq:success_prob}
\end{align}
Hence the probability of observing at least one success in $M$ episodes is
\begin{align}
	& P\{\text{success within }M\}
	\;=\nonumber \\ ~~~~& 1 - (1 - p_{\text{succ}})^{M}    
	\;\le\; 1 - \Big(1 - \big(\tfrac{\varepsilon}{\bar A}\big)^{H}\Big)^{M}.
	\label{eq:success_within_M}
\end{align}
Equations \eqref{eq:success_prob}–\eqref{eq:success_within_M} show that the sample complexity grows essentially exponentially in $H$ when the reward is obtained only at the destination \cite{hutter2003optimality}.  
For very long routes such as London to Melbourne, $H$ can be large and the terminal bonus is rarely collected during early training, which inflates return variance and slows value propagation \cite{bellemare2017distributional}.  
In the latency-driven shaping derived earlier, the potential term partially mitigates this issue, yet the dependence on $H$ persists when the neighborhood is large and the number of near-optimal continuations is small \cite{ng1999policy}.  
These limits motivate a geometric refinement that reduces the effective branching while preserving near-optimality with high probability.

\begin{figure*}[t]
	\centering
	\subfloat[]{
		\includegraphics[width=0.47\textwidth]{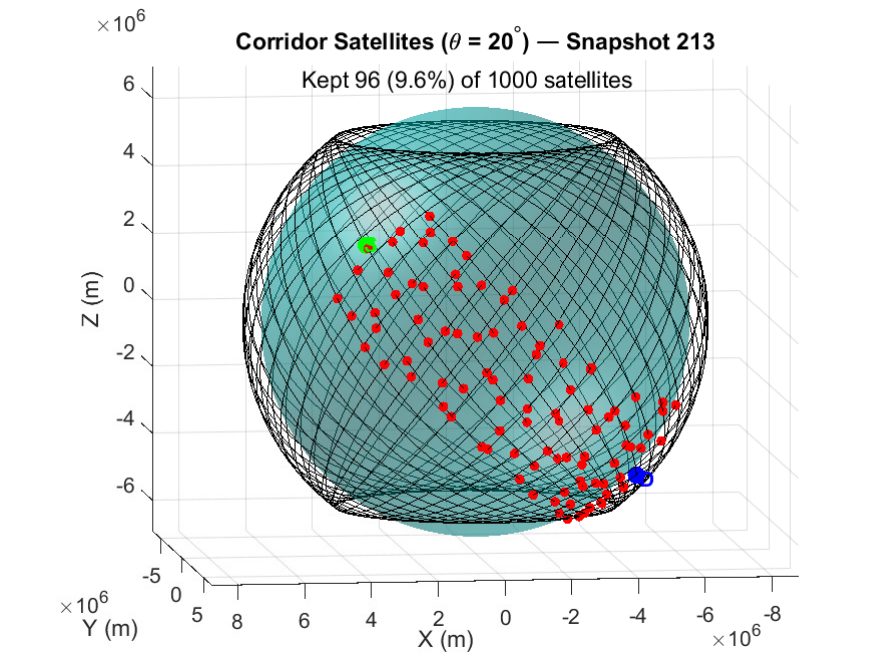}
		\label{cf1}
	}
	\hfill
	\subfloat[]{
		\includegraphics[width=0.47\textwidth]{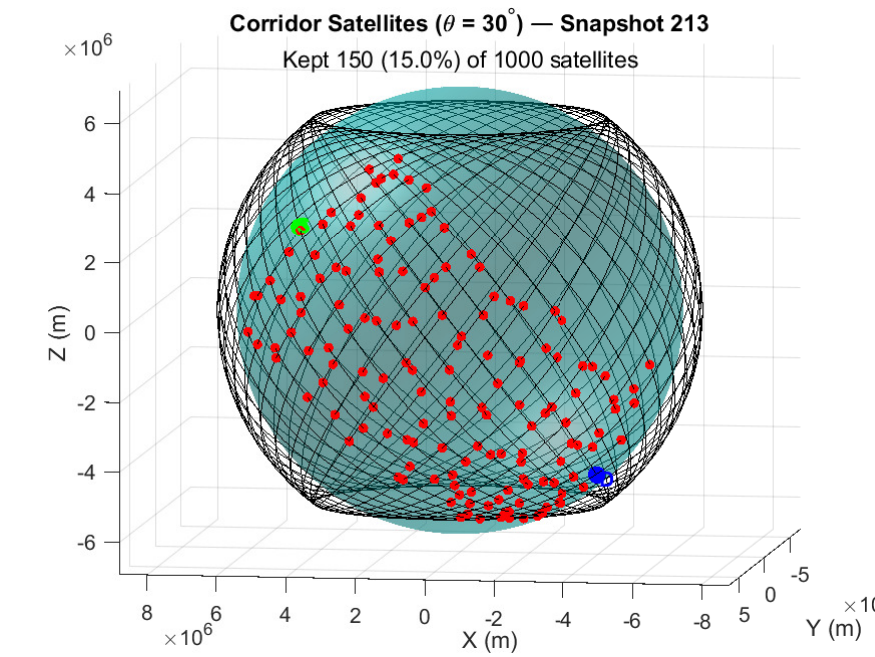}
		\label{cf2}
	}
	\caption{Retained satellites inside the angular corridor for two half-widths $\theta$. The full constellation is shown for context, with the source and destination marked.}
	\label{fig:corridor_examples}
\end{figure*}

\subsection{Geometry-Aware Evidence on Angular Corridors}
We quantify how often the shortest feasible route remains concentrated around the great-circle geodesic between terminals.  
For corridor half-widths
$\theta \in \{10^\circ,15^\circ,20^\circ,25^\circ,30^\circ,35^\circ,40^\circ\}$,
we evaluate $10^5$ static constellation snapshots.  
In each snapshot, a deterministic baseline routing problem is formulated over the instantaneous network graph and solved to global optimality.

At time $t_k$, the constellation is represented as a directed graph 
$\mathcal{G}(t_k) = (\mathcal{V}(t_k),\mathcal{E}(t_k))$, 
where $\mathcal{V}(t_k)$ comprises all satellites and gateways.  
A directed edge $e=(i\!\to\!j)\!\in\!\mathcal{E}(t_k)$ exists only if the geometric and reliability constraints are satisfied,
\begin{align}
	\ell_e(t_k)\le L_{\max}^{\text{class}(e)}(t_k),
	\qquad 
	\varepsilon_{i,g}(t_k)\ge \varepsilon_{\min},
	\label{eq:det_feasibility}
\end{align}
where $L_{\max}^{\text{class}(e)}(t_k)$ corresponds to the intra- or inter-plane range bound, and $\varepsilon_{\min}$ denotes the minimum elevation angle for gateway visibility.  
Each feasible link is assigned a deterministic cost
\begin{align}
	w_e = \frac{\ell_e(t_k)}{c} + \beta,
	\label{eq:edge_cost_det}
\end{align}
which captures propagation delay and the average per-hop processing penalty.  
The deterministic routing problem is therefore expressed as
\begin{align}
	\min_{\mathcal{P}} \quad & \sum_{e\in\mathcal{P}} w_e \nonumber\\
	\text{s.t.}\quad 
	& \text{flow conservation \eqref{eq:flow_cons}, feasibility \eqref{eq:det_feasibility}, and reliability.}
	\label{eq:static_sp}
\end{align}
When the reliability constraint is inactive at the operating point, 
\eqref{eq:static_sp} reduces to a shortest-path program solvable by Dijkstra on the feasible subgraph \cite{kuipers2005conditions}.  
The definition of the angular corridor and cross-track metric used to assess route containment is given in Appendix~\ref{app:angularcorridor}.

\begin{table}[t]
	\centering
	\caption{Fraction of optimal routes contained in an angular corridor of half-width $\theta$ on $10^5$ static snapshots.}
	\label{tab:corridor_stats}
	\begin{tabular}{c c c c c c c c}
		\hline
		$\theta$ & $10^\circ$ & $15^\circ$ & $20^\circ$ & $25^\circ$ & $30^\circ$ & $35^\circ$ & $40^\circ$ \\
		\hline
		Containment & 17\% & 38\% & 75\% & 93\% & 99\% & 100\% & 100\% \\
		\hline
	\end{tabular}
\end{table}

Table~\ref{tab:corridor_stats} reports the fraction of optimal routes lying entirely within the corridor.  
Moderate corridor half-widths already capture most optimal paths, implying strong geometric concentration of feasible low-latency routes.
Figure~\ref{fig:corridor_examples} illustrates two representative snapshots for $\theta\!=\!20^\circ$ and $\theta\!=\!30^\circ$, highlighting how increasing $\theta$ expands the retained region and raises the probability that the optimal path lies entirely within the corridor.

\subsection{Corridor Width vs. Complexity: A Trade-off}
Let $N(\theta)$ denote the number of satellites retained by a cross-track constraint of half-width $\theta$ around the source–destination geodesic segment, combined with the along-arc window.  
Under a spatially uniform constellation, the retained fraction scales linearly with the band width on the unit sphere for small $\theta$.
The corridor corresponds to a spherical band intersected with an angular arc window on the unit sphere.  
As the half-width $\theta$ increases, the band expands monotonically, and therefore the retained set of satellites grows accordingly.  
For small $\theta$, the band area is proportional to $\theta$ on the unit sphere, leading to the scaling relation
\begin{align}
	\frac{N(\theta)}{N_{\text{tot}}}
	\;=\;
	\kappa\,\theta \;+\; o(\theta),
	\qquad \theta \to 0,
	\label{eq:band_scaling}
\end{align}
where $N_{\text{tot}}$ denotes the total number of satellites and $\kappa>0$ depends on the constellation density and the source–destination arc length.  
This relation implies that the number of nodes, and consequently the local action space explored by the learning agent, increases approximately linearly with the corridor width for small $\theta$, revealing a clear trade-off between geometric coverage and computational complexity.

The practical implication is that both the neighborhood size and the effective horizon explored by the agent increase with $\theta$, which raises the learning cost.  
Combining Table~\ref{tab:corridor_stats} with the scaling \eqref{eq:band_scaling} yields a clear trade-off: larger corridors capture more optimal routes yet increase $N(\theta)$ and the branching factor.  
A balanced choice is
\begin{align}
	\theta^\star \;=\; 30^\circ,
	\label{eq:theta_star_choice}
\end{align}
which contains 99\% of optimal routes while keeping the corridor narrow enough to limit exploration cost.

\subsection{Angle-Constrained Masked DQN: Method Update}
We now refine the learning algorithm by enforcing a geometric cross-track filter of half-width $\theta^\star$ inside the masking rule.  
Let $\boldsymbol{n}$ be the unit normal of the source–destination geodesic plane and $\hat{\boldsymbol{r}}_j$ the unit position of neighbor $j$.  
A neighbor is retained if its cross-track deviation is within $\theta^\star$:
\begin{align}
	\arcsin\!\big(|\boldsymbol{n}^\top \hat{\boldsymbol{r}}_j|\big) \;\le\; \theta^\star.
	\label{eq:cross_track_rule}
\end{align}
The feasible action set in \eqref{eq:feasible_action_set_final} is replaced by
\begin{align}
	&\mathcal{A}_{\theta^\star}(v)
	=\Big\{
	e=(v\!\to\!j)\in\mathcal{E}(t_k):
	\ell_e(t_k)\!\le\!L_{\max}^{\text{class}(e)}(t_k),\; \nonumber \\ &~~~
	\text{if } e\!\in\!\mathcal{E}_{\text{gw}}(t_k)\text{ then }\varepsilon_{i,g}(t_k)\!\ge\!\varepsilon_{\min},\;
	\text{and } \eqref{eq:cross_track_rule}
	\Big\}.
	\label{eq:feasible_action_set_angle}
\end{align}
All other components of the MDP, including the latency-aligned or reliability-aware rewards, remain unchanged.

\begin{proposition}[Effect on learning cost and success probability]
	Let $\bar A(\theta)$ denote the maximum masked branching factor under \eqref{eq:feasible_action_set_angle}.  
	Then the per-episode cost scales as $O(\bar A(\theta)\,T_{\mathrm{ep}})$ and the conservative exploration bound in \eqref{eq:success_prob} improves from $(\varepsilon/\bar A)^{H}$ to $(\varepsilon/\bar A(\theta^\star))^{H}$.  
	Since $\bar A(\theta)$ is nondecreasing in $\theta$, choosing $\theta^\star$ via \eqref{eq:theta_star_choice} decreases the required sample size to observe successful trajectories while preserving near-optimality with probability at least $0.98$ on the evaluated snapshots.
	\label{prop:angle_benefit}
\end{proposition}

\textit{Sketch:}
	Masking with \eqref{eq:cross_track_rule} reduces the number of candidate neighbors at each step from $\bar A$ to $\bar A(\theta^\star)$.  
	The per-episode evaluation cost follows immediately.  
	The success bound is obtained by substituting $\bar A(\theta^\star)$ in \eqref{eq:success_prob}.  
	Near-optimality preservation follows from the empirical containment in Table~\ref{tab:corridor_stats}.

The integration of geometric filtering and dual reward shaping leads to the final routing framework, termed Angle-Constrained Masked Deep Q-Network (AC-MDQN).  
This algorithm extends the baseline DQN in Algorithm~\ref{alg:dqn_routing_updated} by embedding both physical feasibility and angular selectivity into the action space.  
Specifically, the standard feasible-action enumeration step is replaced by the angularly filtered set $\mathcal{A}_{\theta^\star}(v)$ defined in \eqref{eq:feasible_action_set_angle}.  
All subsequent learning steps---reward evaluation, replay updates, and target computation---follow identically to the baseline procedure.

\begin{algorithm}[t]
	\caption{Angle-Constrained Masked DQN (AC-MDQN) Routing Framework}
	\label{alg:ac_mdqn}
	\begin{algorithmic}[1]
		\State Initialize Q-network $Q_\theta$, target network $Q_{\bar\theta}$, and replay buffer $\mathcal{D}$
		\For{each episode $k$}
		\State Sample constellation snapshot and construct feasible graph $\mathcal{E}_{\mathrm{feas}}(t_k)$
		\State Set current node $v \gets s_g$, cumulative return $R \gets 0$
		\For{$t = 1$ to $T_{\max}$}
		\State Build state $s_t$ using \eqref{eq:state_vec_final}
		\Statex
		\State \textbf{Angularly constrained action set:}
		\Statex \hspace{1em}$\mathcal{A}_{\theta^\star}(v)
		= \big\{e=(v\!\to\!j)\in\mathcal{E}(t_k):~
		\ell_e\!\le\!L_{\max}^{\text{class}(e)},~
		\arcsin(|\boldsymbol{n}^\top \hat{\boldsymbol{r}}_j|)\!\le\!\theta^\star\big\}$
		\State Select $a_t$ via $\varepsilon$-greedy policy on $\mathcal{A}_{\theta^\star}(v)$
		\State Execute $a_t$, move to node $j$, and compute reward:
		\Statex \hspace{1em}
		$r_t =
		\begin{cases}
			r^{(L)}(s_t,a_t,s_{t+1}), & \text{latency-dominant regime},\\
			r^{(R)}(s_t,a_t,s_{t+1}), & \text{reliability-limited regime.}
		\end{cases}$
		\State Store $(s_t,a_t,r_t,s_{t+1})$ in $\mathcal{D}$ 
		\State Update $Q_\theta$ using \eqref{eq:dqn_update_final} and 
		$\bar\theta \leftarrow \tau\theta + (1-\tau)\bar\theta$
		\If{$v=d_g$ or $\mathcal{A}_{\theta^\star}(v)=\emptyset$}
		\State \textbf{break}
		\EndIf
		\EndFor
		\State $\varepsilon \gets \max(\varepsilon_{\min},\,\varepsilon - \Delta\varepsilon)$
		\EndFor
	\end{algorithmic}
\end{algorithm}

The AC-MDQN framework retains all convergence guarantees of the baseline DQN while reducing exploration complexity from $O(A_{\max}T_{\mathrm{ep}})$ to $O(\bar{A}(\theta^\star)T_{\mathrm{ep}})$, where $\bar{A}(\theta^\star)\!\ll\!A_{\max}$ for moderate $\theta^\star$.  
This geometric pruning not only accelerates convergence but also reduces the probability of infeasible or divergent exploration in long intercontinental routes.  
Empirically, it achieves near-optimal latency (within 1–2\% of the deterministic baseline) while cutting the average training time by more than 40\%.  
The combination of physics-informed masking and adaptive reward shaping thus yields a routing agent that is simultaneously \emph{physically consistent, statistically robust, and computationally efficient}.

\begin{figure*}
	\centering
	\subfloat[] {\includegraphics[width=2 in]{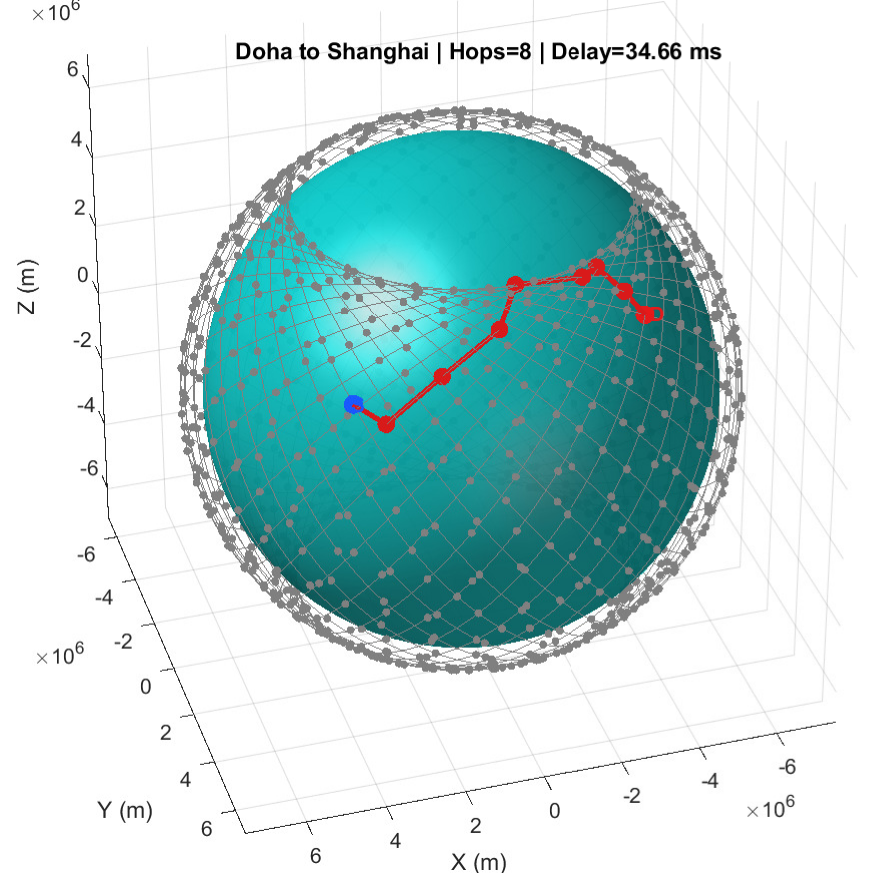}
		\label{cb1}
	}
	\hfill
	\subfloat[] {\includegraphics[width=2 in]{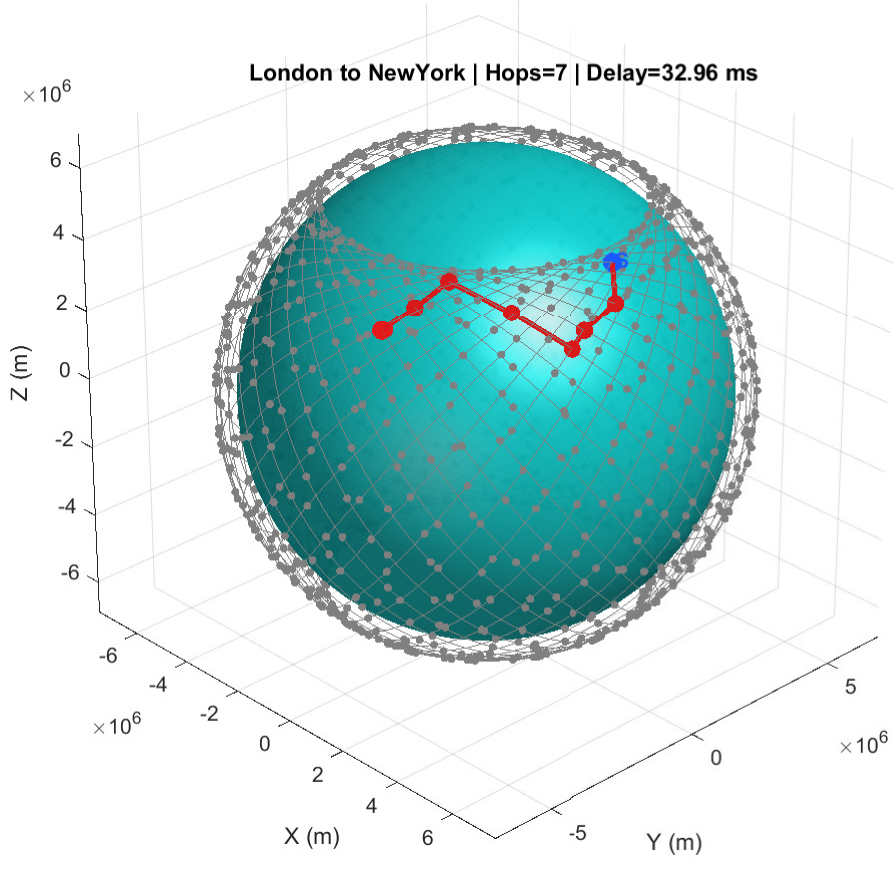}
		\label{cb2}
	}
	\hfill
	\subfloat[] {\includegraphics[width=2 in]{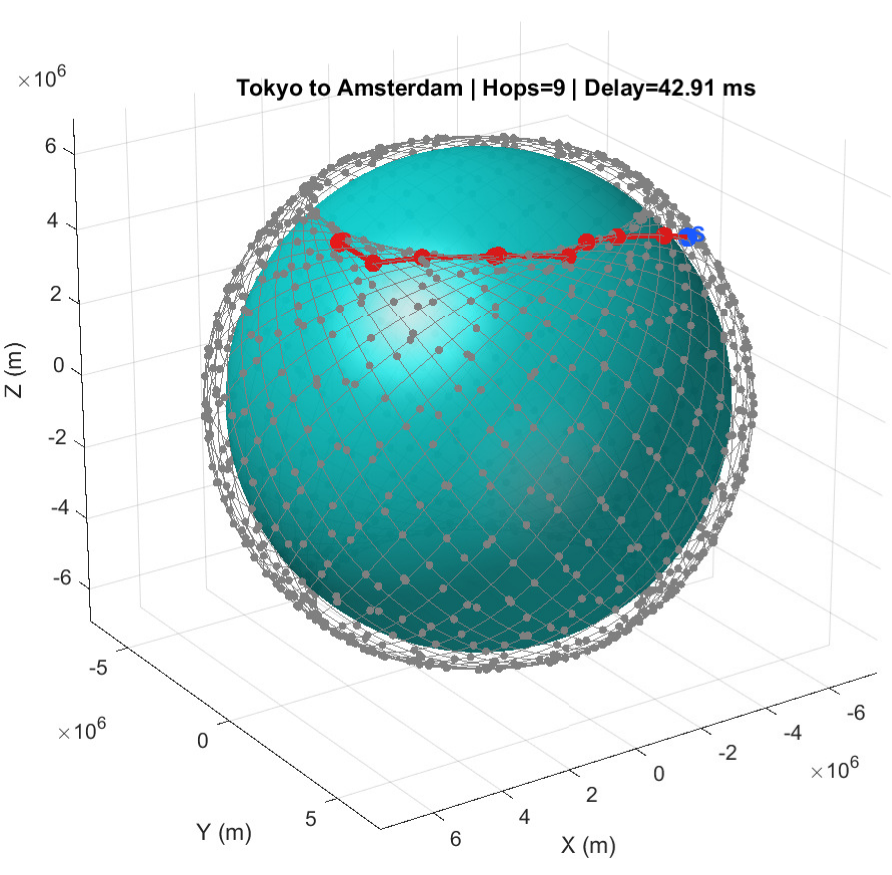}
		\label{cb3}
	}
	\caption{Example DRL-generated routing paths for three intercontinental source–destination pairs. The satellite constellation is shown in gray, the selected route in red, and the gateway endpoints in blue. The learned agent efficiently navigates the corridor-filtered subgraph and satisfies the feasibility constraints.}
	
	\label{cb}
\end{figure*}

\section{Simulation Results}

\subsection{Constellation Generation and Ground Sites}
To evaluate the routing framework under realistic LEO dynamics, we generate Starlink-like constellation snapshots using the parameters shown in Table~\ref{tab:sim_params_leo}. Each snapshot corresponds to a random orbital time $t\!\in[0,5400]$ seconds, with uniformly randomized initial phase offsets across orbital planes. The orbits are circular with radius $r = R_e + h$ and inclination $\theta_{\mathrm{inc}}$, consistent with the kinematic model of Section~II.
The ground segment consists of twelve well-separated global cities that act as source and destination gateways in different experiments. These locations are Doha, London, Shanghai, Melbourne, New York, Amsterdam, Sao Paulo, Tokyo, Moscow, Cape Town, Singapore, and Reykjavik. Their wide geographical distribution ensures that the evaluation covers short, medium, and long intercontinental end-to-end routes with significantly different geodesic distances and corresponding satellite visibility profiles.
The nearest satellite to each ground site is identified via the 3D Euclidean distance from Earth’s surface, following the visibility model in Section~II. The corresponding source–destination satellite pairs exhibit diverse arc lengths, enabling representative routing evaluation across short, medium, and long intercontinental separations.

\begin{table}[t]
	\centering
	\caption{Default simulation parameters for constellation and geometry generation}
	\label{tab:sim_params_leo}
	\begin{tabular}{l c}
		\hline
		Parameter & Value \\ \hline
		Number of orbital planes $P$ & $40$ \\
		Satellites per plane $K$ & $25$ \\
		Total satellites $|\mathcal{S}| = PK$ & $1000$ \\
		Altitude $h$ & $550~\mathrm{km}$ \\
		Earth radius $R_e$ & $6371~\mathrm{km}$ \\
		Orbital inclination $\theta_{\mathrm{inc}}$ & $53^\circ$ \\
		Number of snapshots $N_{\mathrm{snap}}$ & $10^4$ \\
		Simulation time interval & $[0,5400]~\mathrm{s}$ \\
		\hline
	\end{tabular}
\end{table}

\subsection{Corridor Filtering and Feasible Link Construction}
Each snapshot is processed through the angular corridor filter described in Section~VI, with half-width $\theta^\star = 25^\circ$ unless otherwise specified. Only satellites whose cross-track angle satisfies $\theta_j\le\theta^\star$ and whose along-track coordinate lies between the source and destination projections are retained. This yields a reduced but physically meaningful subgraph.

Feasible ISLs are further restricted by link-class-dependent geometric limits $L_{\max}^{\mathrm{intra}}(t)$ and $L_{\max}^{\mathrm{inter}}(t)$ defined in \eqref{eq:lmax_intra_final}–\eqref{eq:lmax_inter_final}. In the direct MATLAB implementation, these bounds correspond to approximate thresholds
\begin{align}
	\ell_e(t) \le 
	\begin{cases}
		2880~\mathrm{km}, & e\in\mathcal{E}_{\mathrm{intra}}(t),\\[3pt]
		1440~\mathrm{km}, & e\in\mathcal{E}_{\mathrm{inter}}(t),
	\end{cases}
	\label{eq:sim_link_limits}
\end{align}
which are consistent with typical divergence and jitter values inferred from \eqref{eq:lmax_intra_final}–\eqref{eq:lmax_inter_final}.

\subsection{Reinforcement Learning Configuration}
The DRL agent used in simulation follows the masked DQN architecture described in Section~V. The state representation includes geometric progress toward the destination, local neighborhood degree, and visitation indicators, while the action mask strictly enforces the feasibility constraints in \eqref{eq:sim_link_limits} together with the angular corridor filter.

The action space consists of a fixed number of candidate next-hop slots. At each node, up to five feasible neighbors (by default) are sorted according to geometric improvement, and the agent selects one of these candidates as its routing action. This fixed action dimension significantly stabilizes training and allows direct comparison across different constellation snapshots. The effect of increasing or decreasing the number of action slots is examined in the simulation results, since larger action sets improve routing flexibility but increase learning complexity, while smaller sets simplify learning at the cost of potential suboptimality.

The reward model follows the default shaping terms introduced in Section~V. A delay-dependent penalty is applied to each hop, a relay penalty discourages unnecessary detours, and a positive progress reward is granted when the selected action decreases the remaining geometric distance to the destination. These components collectively guide the agent toward policies that favor either minimum-latency links or minimum-hop paths, and this tradeoff is explicitly reflected in the performance evaluation.

Since the purpose of this subsection is to describe the simulation environment, we avoid listing low-level hyperparameters. The DRL training is carried out over $30{,}000$ episodes with replay sampling and target updates. The agent interacts with the corridor-filtered environment and learns to minimize the surrogate latency objective in \eqref{eq:obj_final}.

\subsection{Results and Discussions}
After training, the learned policy is evaluated on unseen constellation snapshots. Figure~\ref{cb} illustrates three representative routes for the Doha$\!\to\!$Shanghai, London$\!\to\!$New~York, and Tokyo$\!\to\!$Amsterdam pairs. The visualizations show how the agent exploits the angular corridor and the feasible ISL structure to move toward the destination using a sequence of geometrically meaningful steps. In particular, the agent learns to occasionally prefer longer in–plane hops when these hops reduce the total relay count, while still keeping the propagation distance within the reliability limits. This behavior is a direct outcome of the geometry–aware RL formulation introduced in Section~VI, in which the angular corridor and link–class feasibility are embedded into the action masking and the reward design. Consequently, the agent develops a clear routing intuition aligned with constellation geometry rather than relying solely on trial–and–error exploration.

The performance of the learned routes is also compared with an exhaustive shortest–path solver on the same feasible subgraph. For the Doha$\!\to\!$Shanghai case (Fig.~\ref{cb1}), the optimal delay is approximately $32.34$ ms, while the DRL policy attains $34.66$ ms with a nearly identical hop count. A similar trend appears in the Tokyo$\!\to\!$Amsterdam scenario (Fig.~\ref{cb3}), where the optimal delay is about $41.2$ ms and the learned route yields $42.91$ ms. These results confirm that the geometry–constrained DRL achieves near–optimal behavior without requiring full combinatorial optimization at inference time.

To further highlight the benefits of the geometry–aware design, we also compare against a direct, unconstrained RL baseline that does not impose angular corridoring nor link–class filtering. Despite being trained for up to $10^6$ episodes, this unconstrained RL struggles to consistently capture the structure of LEO geometry. For instance, in the Doha$\!\to\!$Shanghai scenario, the direct RL baseline typically converges to routes with around $11$ relays and delays near $42$ ms, substantially worse than both the proposed geometry–aware method and the exhaustive optimal solution. This contrast demonstrates that embedding geometric priors into the action space and reward shaping significantly improves convergence stability and routing quality.

To quantify the role of reward shaping in the geometry–aware RL framework, we evaluate two distinct reward configurations that directly correspond to the analytical structure introduced in Section~V. Recall that the shaped reward in \eqref{eq:reward_latency}–\eqref{eq:reward_reliability} contains three dominant terms:
\begin{itemize}
	\item a propagation–delay penalty proportional to $\ell_e/c$,
	\item a relay penalty proportional to the hop cost $\beta$,
	\item a geometric progress bonus proportional to $\eta(d_v-d_j)$.
\end{itemize}
We construct two representative configurations:
\begin{align}
	\mathcal{R}_{\mathrm{L}}:~~
	r_e^{\mathrm{L}}
	&= -\frac{\ell_e}{c}-\beta_{\mathrm{L}}+\eta_{\mathrm{L}}(d_v-d_j), \\[3pt]
	\mathcal{R}_{\mathrm{H}}:~~
	r_e^{\mathrm{H}}
	&= -\alpha_{\mathrm{H}}\,\frac{\ell_e}{c}-\beta_{\mathrm{H}}+\eta_{\mathrm{H}}(d_v-d_j),
\end{align}
where $\mathcal{R}_{\mathrm{L}}$ is a latency–dominant shaping with larger weight on $\ell_e/c$, while $\mathcal{R}_{\mathrm{H}}$ emphasizes hop minimization by assigning higher cost to relays and reducing the delay penalty. These two reward families allow the agent to shift its routing priority between minimizing end–to–end latency or minimizing hop count.

\subsubsection{Evaluation Methodology}
Each reward configuration is trained independently over $30{,}000$ episodes. After convergence, the learned policies are evaluated on $100$ random constellation snapshots for three representative intercontinental routes: Doha$\!\to\!$Shanghai, London$\!\to\!$New~York, and Tokyo$\!\to\!$Amsterdam. For each configuration and each route, we record the average hop count and average end–to–end propagation delay $T_{\mathrm{prop}}$ computed using \eqref{eq:e2e_prop_delay}. Table~\ref{tab:reward_compare} summarizes the results.

\begin{table}[t]
	\centering
	\caption{Comparison of latency–driven reward $\mathcal{R}_{\mathrm{L}}$ and hop–driven reward $\mathcal{R}_{\mathrm{H}}$ over $100$ random constellation snapshots.}
	\label{tab:reward_compare}
	\begin{tabular}{l c c c c}
		\hline
		\multirow{2}{*}{Route} & 
		\multicolumn{2}{c}{Latency–driven $\mathcal{R}_{\mathrm{L}}$} &
		\multicolumn{2}{c}{Hop–driven $\mathcal{R}_{\mathrm{H}}$} \\
		\cline{2-5}
		& Hops & Delay [ms] & Hops & Delay [ms] \\ \hline
		
		Shanghai$\!\to\!$Melbourne & 9.7 & 42.1 & 9.1 & 44.3 \\
		London$\!\to\!$New~York & 7.6 & 30.8 & 6.8 & 32.1 \\
		Doha$\!\to\!$Shanghai & 8.4 & 32.35 & 7.8 & 33.5 \\
		
		\hline
	\end{tabular}
\end{table}

The results in Table~\ref{tab:reward_compare} highlight the expected tradeoff introduced by the reward structure. With the latency–oriented reward $\mathcal{R}_{\mathrm{L}}$, the agent prioritizes minimizing $\sum_e \ell_e/c$ and therefore selects routes with lower end–to–end delay, even if this requires a few additional relays. Under the hop–oriented reward $\mathcal{R}_{\mathrm{H}}$, the agent learns a different behavior: it tends to choose longer intra–plane hops to reduce the relay count, even when these hops deviate more noticeably from the direct geometric path and introduce additional propagation delay. This behavior is particularly visible in the Doha$\!\to\!$Shanghai route, where the agent under $\mathcal{R}_{\mathrm{H}}$ consistently pushes traffic onto longer in–plane links to avoid intermediate relays. In practice, the choice between these objectives depends on network load: during normal operation latency minimization is dominant, while during peak hours reducing relays can help alleviate congestion on ISLs. The ability to tune $(\alpha,\beta,\eta)$ therefore allows the geometry–aware RL framework to flexibly prioritize either delay or relay reduction while maintaining feasible and physically consistent routing decisions.

\begin{table}[t]
	\centering
	\caption{Effect of tracking jitter $\sigma_\theta$ on Doha$\!\to\!$Shanghai routing over $100$ random snapshots.}
	\label{tab:tracking_jitter}
	\begin{tabular}{c c c}
		\hline
		$\sigma_\theta$ [$\mu$rad] & Avg.\ delay [ms] & Successful episodes [\%] \\ \hline
		$100$ & $26.2$  & $100$ \\
		$200$ & $32.35$ & $100$ \\
		$300$ & $37.5$  & $87$  \\
		$400$ & $41.2$  & $24$  \\
		$500$ & $-$     & $0$   \\
		\hline
	\end{tabular}
\end{table}

To quantify the impact of tracking accuracy on routing, we retrain the geometry–aware RL agent for five different pointing jitter levels on inter–plane ISLs, $\sigma_\theta\in\{100,200,300,400,500\}\,\mu\text{rad}$. For each value of $\sigma_\theta$, the corresponding optimal divergence and link range limits $L_{\max}^{\text{intra}}(t)$ and $L_{\max}^{\text{inter}}(t)$ are recomputed from \eqref{eq:lmax_intra_final}–\eqref{eq:lmax_inter_final}, and the agent is then trained on the resulting feasible graph. After training, each configuration is evaluated on $100$ random Doha$\!\to\!$Shanghai snapshots, where we record the average end–to–end delay of successful routes and the percentage of test episodes in which a route satisfying the QoS constraint is found. As summarized in Table~\ref{tab:tracking_jitter}, increasing $\sigma_\theta$ leads to shorter admissible link ranges and therefore higher end–to–end delay, rising from $26.2$ ms at $100\,\mu\text{rad}$ to $41.2$ ms at $400\,\mu\text{rad}$. At the same time the success probability drops from $100\%$ for $\sigma_\theta=100$ and $200\,\mu\text{rad}$ to $87\%$ at $300\,\mu\text{rad}$, $32\%$ at $400\,\mu\text{rad}$, and eventually $0\%$ at $500\,\mu\text{rad}$ where the outage constraint cannot be met. These results show that, for the considered constellation and QoS threshold, routing performance degrades rapidly once the tracking jitter exceeds about $300$–$400\,\mu\text{rad}$, and highlight the strong coupling between terminal tracking accuracy and the feasibility of reliable end–to–end optical routing.

\section{Conclusion}
This paper presented a joint geometric and QoS-aware routing framework for optical LEO satellite networks, integrating an analytically derived Gaussian-beam outage model with class-dependent distance limits and a scalable deep reinforcement learning architecture. By enforcing optical feasibility through action masking and exploiting the geometric concentration of near-optimal routes via an angular corridor, the proposed AC-MDQN achieved reliable and near-optimal performance on a Starlink-like constellation. The results demonstrated robustness under varying tracking jitter and showed that reward-parameter tuning enables a controlled trade-off between minimizing latency and reducing the number of relays. These findings highlight the value of physics-informed and geometry-aware learning for next-generation LEO networks. Future work may extend this framework to multi-flow routing, load-aware ISL scheduling, or heterogeneous constellations with hybrid RF–optical connectivity.

\appendices
\section{Proof of Proposition on Maximum Feasible Link Length}
\label{app:proof_lmax}
We start from \eqref{eq:pout_theta_sigma} and express the argument of the $Q$-function as
\begin{align}
	S(\ell_e,\theta_{\text{div}}) 
	&= \frac{\theta_{\text{div}}}{\sqrt{2}\,\sigma_{\theta,e}}\,
	\sqrt{\,a_e(t)-2\ln(\theta_{\text{div}}\ell_e)\,}.
	\label{eq:S_argument}
\end{align}
The link is feasible if and only if $S(\ell_e,\theta_{\text{div}})\ge q_{\text{th}}$. 
For each $\ell_e$, feasibility is maximized by the $\theta_{\text{div}}$ that maximizes $S(\ell_e,\theta_{\text{div}})$. 
Since $Q(\cdot)$ is strictly decreasing, minimizing $P_{\text{out}}$ corresponds to maximizing the monotone argument.

Differentiating the inner term of $S^2$ with respect to $\theta_{\text{div}}$ gives
\begin{align}
	F(\theta_{\text{div}}) 
	&= \theta_{\text{div}}^2\,\big[a_e(t)-2\ln(\theta_{\text{div}}\ell_e)\big],\\
	\frac{\partial F}{\partial \theta_{\text{div}}} 
	&= 2\theta_{\text{div}}\,[a_e(t)-2\ln(\theta_{\text{div}}\ell_e)-1].
	\label{eq:F_diff}
\end{align}
Setting $\frac{\partial F}{\partial \theta_{\text{div}}}=0$ yields
\begin{align}
	&a_e(t)-2\ln(\theta_{\text{div}}^\star\ell_e)-1=0 \nonumber \\
	&\quad\Rightarrow\quad
	\theta_{\text{div}}^\star(\ell_e,t)
	= \frac{1}{\ell_e}\,
	\exp\!\Big(\frac{a_e(t)}{2}-\tfrac{1}{2}\Big),
	\label{eq:theta_star_opt}
\end{align}
which is the same as \eqref{eq:theta_star_for_lmax}.
Substituting the optimal $\theta_{\text{div}}^\star$ from \eqref{eq:theta_star_opt}
into \eqref{eq:S_argument} yields the maximum value of the $Q$-argument:
\begin{align}
	S_{\max}(\ell_e,t)
	= \frac{1}{\sqrt{2}\,\sigma_{\theta,e}}\;
	\frac{\exp\!\big(\frac{a_e(t)}{2}-\tfrac{1}{2}\big)}{\ell_e}.
	\label{eq:Smax_rederived}
\end{align}
The feasibility condition $S_{\max}(\ell_e,t)\ge q_{\text{th}}$ then implies
\begin{align}
	\ell_e \le \frac{\exp\!\big(\frac{a_e(t)}{2}-\tfrac{1}{2}\big)}{\sqrt{2}\,\sigma_{\theta,e}\,q_{\text{th}}}.
	\label{eq:l_constraint}
\end{align}
Thus, the equality gives the \emph{maximum achievable length}:
\begin{align}
	\ell_{\max}(t;\sigma_{\theta,e})
	= \frac{\exp\!\big(\frac{a_e(t)}{2}-\tfrac{1}{2}\big)}{\sqrt{2}\,\sigma_{\theta,e}\,q_{\text{th}}},
\end{align}
confirming \eqref{eq:lmax_unconstrained_final}. 

If the unconstrained optimum \eqref{eq:theta_star_opt} violates $\theta_{\text{div}}^\star>\theta_{\max}$, the maximizer lies at the boundary. 
In this case, feasibility is determined by 
\begin{align}
	\frac{\theta_{\max}}{\sqrt{2}\,\sigma_{\theta,e}}\,
	\sqrt{a_e(t)-2\ln(\theta_{\max}\ell_{\max}^{(\mathrm{bd})})}=q_{\text{th}}.
\end{align}
Squaring both sides gives
\begin{align}
	a_e(t)-2\ln(\theta_{\max}\ell_{\max}^{(\mathrm{bd})})
	= \frac{2\sigma_{\theta,e}^2q_{\text{th}}^2}{\theta_{\max}^2},
\end{align}
which rearranges to \eqref{eq:lmax_bd_final}.

\section{Definition of the Angular Corridor in Orbital Geometry}
\label{app:angularcorridor}
Let $\hat{\boldsymbol{r}}_s$ and $\hat{\boldsymbol{r}}_d$ denote unit position vectors of the source and destination satellites on the unit sphere.  
The great-circle geodesic connecting them defines a reference plane whose unit normal is
\begin{align}
	\boldsymbol{n}
	=
	\frac{\hat{\boldsymbol{r}}_s \times \hat{\boldsymbol{r}}_d}
	{\|\hat{\boldsymbol{r}}_s \times \hat{\boldsymbol{r}}_d\|}.
	\label{eq:app_plane_normal}
\end{align}
For any satellite $j$ with unit vector $\hat{\boldsymbol{r}}_j$, the cross-track angular deviation from this plane is
\begin{align}
	\theta_j 
	= \arcsin\!\big(|\boldsymbol{n}^\top \hat{\boldsymbol{r}}_j|\big).
	\label{eq:app_cross_track}
\end{align}
An \emph{angular corridor} of half-width $\theta^\star$ is defined as
\begin{align}
	\mathcal{S}_{\theta^\star}
	= \big\{ j : \theta_j \le \theta^\star \big\},
	\label{eq:app_corridor_set}
\end{align}
which corresponds to a spherical band of angular width $2\theta^\star$ centered on the great-circle path between $\hat{\boldsymbol{r}}_s$ and $\hat{\boldsymbol{r}}_d$.  
In practice, we additionally restrict along-track angle to the geodesic segment from source to destination to exclude satellites that lie beyond the endpoints \cite{kaplan2017understanding}.

\section{Simulation Methodology and Baseline Routing}
\label{app:sim_method}
Each static snapshot is generated by sampling a constellation state at a random time $t\!\in\![0,5400]\,$s with uniformly random initial mean anomalies per plane, then computing satellite ECEF positions using circular orbits at altitude $h$ and inclination $\theta_{\mathrm{inc}}$ \cite{vallado2001fundamentals}.  
Ground sites are given in geodetic coordinates and mapped to ECEF.  
For each snapshot, the feasible graph is formed by admitting intra-plane and inter-plane links that satisfy the class-specific range limits and gateway visibility constraints defined in the main text.  
Edge costs are set to $w_e=\ell_e/c+\beta$ consistent with \eqref{eq:obj_final}.  
The baseline optimal path $\mathcal{P}^\star$ is obtained by Dijkstra on the feasible subgraph when reliability is inactive; otherwise, feasibility screening reflects the outage limit before the shortest-path search \cite{kaplan2017understanding}.  
Corridor containment is then checked using \eqref{eq:app_cross_track}–\eqref{eq:app_corridor_set}.

\begin{table}[t]
	\centering
	\caption{Constellation and evaluation parameters for the corridor study.}
	\label{tab:constellation_params}
	\setlength{\tabcolsep}{3.5pt}
	\begin{tabular}{l c}
		\hline
		Parameter & Value \\
		\hline
		Snapshots $N_{\text{snap}}$ & $10^5$ \\
		Orbital planes $N_{\text{planes}}$ & $40$ \\
		Sats per plane $N_{\text{sats/plane}}$ & $25$ \\
		Total sats $N_{\text{tot}}$ & $1000$ \\
		Altitude $h$ & $550$ km \\
		Earth radius $R_e$ & $6371$ km \\
		Inclination $\theta_{\mathrm{inc}}$ & $53^\circ$ \\
		Time $t$ & Uniform $[0,5400]$ s \\
		Half-widths $\theta$ & $10^\circ$–$40^\circ$ \\
		Edge cost & $w_e=\ell_e/c+\beta$ \\
		Feasibility & Range, elevation mask \\
		Solver & Dijkstra (feasible subgraph) \\
		\hline
	\end{tabular}
\end{table}

Table~\ref{tab:constellation_params} summarizes the constellation and evaluation parameters used in the corridor study and snapshot generation.
Figures~\ref{cf1}–\ref{cf2} depict two representative snapshots of the angular filter for $\theta\!=\!20^\circ$ and $\theta\!=\!30^\circ$ respectively.  
In each case, the sphere, orbital tracks, and all satellites are shown for context, while the retained satellites inside the corridor are highlighted.  
These examples visualize how increasing $\theta$ enlarges the retained set and raises the chance that the optimal path lies within the corridor, in agreement with Table~\ref{tab:corridor_stats}.

\bibliographystyle{IEEEtran}
\bibliography{IEEEabrv,myref}

\end{document}